\documentclass[11pt]{article}
\usepackage[boxruled]{algorithm2e}
\usepackage{amsmath, amssymb,amsthm}
\usepackage{graphicx}

\newsavebox{\fmbox}

\newtheorem{theorem}{Theorem}[section]
\newtheorem{lemma}[theorem]{Lemma}

\newtheorem{definition}[theorem]{Definition}
\newtheorem{property}[theorem]{Property}
\newtheorem{observation}[theorem]{Observation}
\newtheorem{claim}[theorem]{Claim}

\makeatletter
\def\fnum@figure{{\bf Figure \thefigure}}
\def\fnum@table{{\bf Table \thetable}}

\long\def\@mycaption#1[#2]#3{\addcontentsline{\csname
 ext@#1\endcsname}{#1}{\protect\numberline{\csname
  the#1\endcsname}{\ignorespaces #2}}\par
     \begingroup
       \@parboxrestore
          \small
       \@makecaption{\csname fnum@#1\endcsname}{\ignorespaces
#3\endgroup}
      }

\title{A Deterministic Algorithm for the Vertex Connectivity Survivable Network Design Problem}
\author{Pushkar Tripathi\footnote{Research Supported by NSF grants CCF-0728640 and CCF-0914732} \\ Georgia Institute of Technology}
\begin{document}
\maketitle

\begin{abstract}
In the vertex connectivity survivable network design problem we are given
an undirected graph $G = (V,E)$ and connectivity requirement $r(u,v)$ for each pair of vertices $u,v \in V$.We are also
given a cost function on the set of edges. Our goal is to find the minimum cost subset of edges such that 
for every pair $(u,v)$ of vertices we have $r(u,v)$ vertex disjoint paths in the graph induced by the chosen edges.
Recently, Chuzhoy and Khanna \cite{khanna} presented a randomized algorithm that achieves a factor of $O(k^3 \log n)$ for this 
problem where $k$ is the maximum connectivity requirement. In this paper we derandomize their algorithm to get 
a \textit{deterministic} $O(k^3 \log n)$ factor algorithm.
Another problem of interest is the single source version of the problem,
where there is a special vertex $s$ and all non-zero connectivity requirements must involve $s$. 
We also give a \emph{deterministic} $O(k^2 \log n)$ algorithm for this problem.
\end{abstract}

\section{Introduction}
In the vertex connectivity survivable network design problem(VC-SNDP) problem 
we are given an undirected graph $G(V,E)$ and a cost function
over the edges, and a connectivity requirement $r(u, v) \leq k$ for all $u, v \in V$ . We are also
given a set $T \subseteq V$ of terminals, and $r(u, v) > 0$ only if both $u$ and $v$ belong to 
$T$. Such pairs of terminals with non-zero connectivity requirements 
are called source-sink pairs. The goal is to find the minimum cost subset of edges 
such that each source sink pair, $(u,v)$ is connected by $r(u,v)$ vertex disjoint paths
in the graph induced by the edges. In the Edge Connectivity SNDP(EC-SNDP) we 
wish to find the minimum cost subset of edges such that the induced paths are edge-disjoint.

The best current approximation algorithm for EC-SNDP is by Jain\cite{kamal1}, using the
technique of iterated rounding.As for the VC-SNDP, no non-trivial algorithm was known 
for this problem until the recent result by Chuzhoy and Khanna\cite{khanna} in which 
they gave an $O(k^3 \log n)$ factor randomized algorithm for this problem. Prior to 
the work of Chuzhoy and Khanna, there had been some progress for some restricted special 
cases of this problem. For the case when $k=1\ or \ 2$, $2$-approximate algorithms were known,
due to Agarwal et. al.\cite{agarwal} and Fleisher \cite{fleischer012}. The $k$-vertex connected spanning
subgraph problem, which is a special case of VC-SNDP where the connectivity requirement between every pair
of vertices is equal to $k$, has been studied extensively. Cheriyan et al. \cite{cherian1,cherian2} gave 
an $O(\log k)$-approximation algorithm for this case when $k$ is at most $\sqrt{n/6}$, and an $O(\sqrt{n/\epsilon})$-
approximation algorithm for $k \leq (1-\epsilon)n$. The bound for large values of $k$
was improved further by a series of improvements \cite{kortsarz,fakch} culminating 
in the current best factor of $O(\log k. \log \frac{n}{n-k})$ due to Nutov\cite{nutovManuscript}.

On the hardness front, Kortsarz et. al \cite{kortsarz2} showed that VC-SNDP is hard to approximate
beyond a factor of $2^{\log ^{1-\epsilon} n}$for any $\epsilon >0$ when $k$ is polynomially large in $n$.
Chakraborty et. al\cite{chakrabarty} extended this result a $k^{\epsilon}$ hardness for $k > k_0$ for some 
fixed positive constants $k_0$ and $\epsilon$. The problem is known to be APX hard even for small values of $k\geq 3$.

In the single source version of the VC-SNDP, there is a special vertex $s$, and all
non-zero connectivity requirements must involve $s$. This problem has received a lot of attention
lately. Kartsarz et.al \cite{kortsarz2} showed that this version of the problem is hard to approximate beyond a factor 
of $\Omega(\log n)$. Lando and Nutov \cite{lando} further improved this to $\log ^{2-\epsilon}n$ for large(polynomial in $n$)
values of $k$. Chakraborty et. al \cite{chakrabarty} gave an $O(2^{O(k^2)} \log^4 n)$ for this problem. Through a series of 
improvements by \cite{chakrabarty,korula1,khanna2,khanna} this factor was brought down to $O(k^2 \log n)$. In \cite{khanna2} 
chuzhoy et. al presented an $O(k\log n)$ algorithm for the special case when 
all non-zero connectivity requirements are equal to $k$.

Element-connectivity SNDP is a closely related problem to EC-SNDP and VC-SNDP which
was first defined by \cite{vijay}. The input to the element-connectivity SNDP is the same as for EC-SNDP and
VC-SNDP, and the set of terminals $T$ is defined as above. We define an element as 
any edge or any non-terminal vertex in the graph. We say that
a pair $(s, t)$ of vertices is $k$-element connected iff $s$ and $t$
 cannot be disconnected by deleting a set of at most $(k - 1)$
elements. As with the EC-SNDP and VC-SNDP we are required to find the 
minimum cost subset of edges such that every pair of vertices,$(u,v)$ in
 the induced is $r(u,v)$-element connected. 

The element connectivity problem is considered to be of intermediate difficulty between the
VC-SNDP and EC-SNDP due to the following observation. If $(s, t)$ are $k$-vertex connected, 
then they are also $k$-element connected, furthermore, if $(s, t)$ are $k$-element 
connected, then they are also $k$-edge connected. In \cite{fjw} Fleischer et. al present a $2$-approximation
algorithm for this problem.

\subsection{Our Results}
In this paper we present a deterministic $O(k^3 \log n)$ factor algorithm for
the VC-SNDP. Our algorithm is based on the technique introduced by Chuzhoy and Khanna\cite{khanna}. 
In their algorithm Chuzhoy and Khanna prove a key technical lemma that uses randomization to show 
the existance of a bipartite graph with certain desireable properties(refer to property \ref{goodness_property}). 
Almost every sufficiently dense random graph satisfies these properties, but it is 
NP-hard to even verify these properties for a given graph. This precludes standard
derandomization techniques such as those introduced by Nissan and Wigderson\cite{nissan} and recently 
by Sivakumar\cite{siva} which rely on space efficient property testing.

Another common approach to derandomization is by the use of conditional expectation techniques like in \cite{cond_derand}.
These techniques rely on progressively replacing random choices by deterministic ones while ensuring the 
probability of failure is less than $1$, if the remaining choices are made 
based on random coin flips. These techniques are difficult to apply for this problem because the failure probability
seems to be very hard to calculate.

We overcome these problems by giving an alternate, \textit{stronger} set of contraints(refer property \ref{strong_goodness_restate})
that characterize the graphs with the desired properties. Our characterization of 
such graphs is simpler and easier to verify and is therefore amenable to derandomization. 
We exhibit a deterministic, polynomial time, local search based algorithm to find a graph satisfying these constraints. 
Our local search algorithm is combinatorial but uses a linear programming formulation of the 
desired properties and probabilistic arguements to reason about its convergence.
We also present a deterministic $O(k^2 \log n)$ factor algorithm for the single source variant of this problem. 


\section{Vertex Connectivity Algorithm}
\label{their_algo}
In this section we present an overview of the $O(k^3 \log n)$ algorithm 
for the VC-SNDP by $\cite{khanna}$ that uses the constant factor approximation algorithm
for the element connectivity SNDP as a subroutine.

We begin by making $m$ copies of the original 
graph say $G_1 \cdots G_m$. Associated with each copy is a subset of
terminals $T_i \subseteq T$. We view each pair $(G_i,T_i)$ as an instance 
of the the element connectivity problem in the following way. For each pair of 
vertices $s,t \in T_i$ the new connectivity requirement is the same as the original one. 
The connectivity requirement for all other pairs is $0$. Note that for each copy the
cost of the optimal solution is at most $OPT$. We then use the $2$-approximation algorithm 
for the element connectivity problem by \cite{fjw} for each of the instances of $k$-element
connectivity SNDP. Let $E_i$ be the solution returned by this algorithm for the $i^{th}$ 
instance. The final solution is the union of the solutions for all the instances, i.e. 
$\bigcup_{i \in [m]} E_i$.
Clearly, the cost of this solution is at most $2m$ times $OPT$. The main idea of the algorithm is that
for an appropriate choice of sets $T_i$, we can guarantee that the solution produced above is a 
feasible solution. We call such a choice of sets $T_i$ as a \textit{good family of subsets}.

\begin{definition}[Good Family of Subsets]
Let \textsl{M} be the input collection of source-sink pairs and $T$ is the corresponding
collection of terminals. We say that a family $\left\{T_1, \cdots , T_m \right\}$ of 
subsets of T is \textit{good} iff for each source-sink pair $(s, t) \in M$, 
for each subset of terminals, $X$, of size at most $(k - 1)$, there
is a subset $T_i$, such that $s, t \in T_i$ and $X\cap T_i  = \phi$.
\end{definition}

We now show that the existance of such a family would imply that the output of the above algorithm is a feasible solution to VC-SNDP.

\begin{theorem}[\cite{khanna}] 
\label{main_thm}
Let $\left\{T_1, \cdots, T_m\right\}$ be a good family of subsets. 
Then the output of the above algorithm is a feasible solution to the VC-SNDP instance.
\end{theorem}
\begin{proof}
Let $(s, t)  \in M$ be any source-sink pair, and let $X \subseteq V / \left\{s, t\right\}$ be any collection
of at most $(r(s, t)-1) \leq (k -1)$ vertices. It is enough to show that the removal of $X$ from
the graph induced by the set of edges returned by the above algorithm does not 
separate $s$ from $t$. Let $X' = X \cap T$. Since $\left\{T_1, \cdots , T_m\right\}$ is a good 
family of subsets, there is some $T_i$ such that $s, t  \in T_i$ while $T_i \cap X  = \phi$. Recall
that set $E_i$ of edges defines a feasible solution to the element-connectivity SNDP instance
corresponding to $T_i$. Then $X$ is a set of non-terminal vertices with respect to $T_i$. Since $s$
is $r(s, t)$-element connected to $t$ in the graph induced by $E_i$, the removal of $X$ from the
graph does not disconnect $s$ from $t$.
\end{proof}

\cite{khanna} gives a randomized algorithm to construct such a good family of subsets of size $O(k^3 \log n)$. 
By theorem \ref{main_thm} this implies an $O(k^3 \log n)$ 
factor algorithm for the VC-SNDP problem. In the following section we present a deterministic algorithm
to construct such a family of size $O(k^3 \log n)$.

\section{Algorithm to Construct Good Family of Subsets}
We may view a good family of subsets as a bipartite graph, $G = (L \cup R, E)$ as follows.
Every vertex $\ell_j \in L$ corresponds to a terminal vertex and each vertex $r_p \in R$ can be identified 
with a subset of terminals $T_p \subseteq T$, where terminal $t_j \in T_p$ iff $(\ell_j, r_p) \in E$. 


In light of the representation above the definition of a good family of subsets may be 
restated as follows. 

\begin{property}[Weak Goodness]
\label{goodness_property}
A graph $G$ satisfies the \textit{weak goodness property} if for every $\ell_i, \ell_j \in L$ and 
$X \subseteq L$ such that $|X| < k$ and $\left\{ \ell_i,\ell_j 
\right\} \notin X$, $N(\left\{\ell_i\right\})\cap N(\left\{\ell_j\right\})$ is not contained in  $N(X)$, 
where $N(S)$ denotes the set of neighbours of $S\subseteq L$.
\end{property}

Our task is to build a graph $G$ which satisfies property \ref{goodness_property} and has $|R|$ as 
small as possible. The main difficulty with the property \ref{goodness_property} is that it has exponentially many
constraints and cannot be checked efficiently. We use an alternate, stronger definition of good family of subsets.

\begin{property}[Strong Goodness]
\label{strong_goodness}
A graph $G = (L \cup R, E)$ satisfies the $(\alpha, \beta)-strong \ goodness$ property if 
\begin{enumerate}
\item for all distinct $\ell_i, \ell_j \in L$,  $|N(\left\{\ell_i\right\})\cap N(\left\{\ell_j\right\})| \geq \alpha$
\item for all distinct $\ell_i, \ell_j, \ell_t \in L$,  $|N(\left\{\ell_i\right\})\cap N(\left\{\ell_j\right\}) \cap N(\left\{\ell_t\right\})| \leq  \beta$
\end{enumerate}
\end{property}

The following observation follows easily from the definitions above.

\begin{observation}
\label{link_properties}
If a graph $G$ satisfies the \textit{$(\alpha,\beta)$-strong goodness} property for $\alpha / \beta  > k$,
 then it satisfies the weak goodness property.
\end{observation}

The main advantage of the definition for strong goodness is that it has polynomially many constraints, which 
allows us to efficiently ascertain whether a given family of subsets(represented by a bipartite graph $G$) 
is a good family or not. This was not possible in the previous definition.
Now we describe a randomized algorithm to construct a graph 
$G$ such that it satisfies the \textit{$(\alpha,\beta)$-strong goodness} property
for $\alpha / \beta > k$. We will subsequently derandomize this algorithm. 
The variables $\alpha,\beta$ will be determined later.

\subsection{Algorithm for Constructing Graphs with Strong Goodness Property}
\label{algo_section}
Our algorithm constructs a graph $G=(L \cup R, E)$, for $|L| = n$ by assigning a string $w_i \in \cal A^{\gamma}$
to every $\ell_i \in L$. There are $|\cal A|\cdot\gamma$ vertices in $R$
which are indexed by $[\gamma] \times \cal A$. There is an edge connecting $\ell_i$ to $(j,c)$ iff 
$w_i[j] = c$, i.e. the $j^{th}$ character of $w_i$ is $c$. Property \ref{strong_goodness} can be restated as follows,

\begin{property}
\label{strong_goodness_restate}
A graph $G$ satisfies the \textit{$(\alpha, \beta)$-strong goodness} property if 
\begin{enumerate}
\item for all distinct $\ell_i, \ell_j \in L$,  $|w_i \odot w_j| \geq \alpha$
\item for all distinct $\ell_i, \ell_j, \ell_t \in L$,  $|w_i \odot w_j \odot w_t| \leq \beta$
\end{enumerate}
Where $|s_1 \odot s_2|$ is the number of places where the strings $s_1$ and $s_2$ agree.
\end{property}

Our algorithm iteratively builds the set of labels for vertices in $L$ by solving a linear program and
rounding its output using randomized rounding. At the end of each iteration we maintain the invariant that the
set of labelled vertices satisfy property \ref{strong_goodness_restate}. Next, we describe the algorithm 
formally, and subsequently find the minimum value of $\gamma$ for which this algorithm will work. 

Let $S_r = \left\{s_1, s_2 \ldots s_r \right\}$ be the set of labels after $r$ iterations. 
Define $S_2 = \left\{\mu, \nu \right\}$ where $\mu^j_c = 1$ for all $j \in \gamma$ for some fixed $c \in {\cal A}$ and 
$\nu = \left[ \left[ c_1 c_2 \cdots c_{|\cal A|} \right]  \cdots \left[ c_1 c_2 \cdots c_{|\cal A|} \right] \cdots \gamma/|{\cal A}| \ times \right]$. In each subsequent iteration we augment 
this set by another label by solving the following integer program.

\begin{eqnarray} 
    (IP1) \hspace{3 cm} \sum_{c \in \cal A} x^j_c &=& 1 \hspace{4.8 cm} \forall j \in [\gamma] \\ 
    \sum_{j \in [\gamma]} \sum_{c \in \cal A} x^j_c \cdot f(s_p,j,c) &\geq& 2\alpha  \hspace{4.6 cm} \forall p \in S_r \\
    \sum_{j \in [\gamma]} \sum_{c \in \cal A} x^j_c \cdot g(s_p, s_q,j, c) &\leq& 2\beta/3  \hspace{4.2 cm} \forall p,q\in S_r\\
    x^j_c &\in& \left\{ 0,1\right\} \hspace{4.0 cm} \forall j \in [\gamma], c \in \cal A
\end{eqnarray}

Here $f(s,j,c)$ returns $1$ if the $j^{th}$ character of the string $s$ is $c$ and it returns zero otherwise. Similarly $g(s_1,s_2, i,c)$ returns $1$
if the $j^{th}$ character of both strings $s_1$ and $s_2$ is $c$, and returns zero otherwise. We interpret the above 
program as follows. If the variable $x^j_c$ is set to $1$ it implies that the $j^{th}$ character of the new string($s_{r+1}$) is $c$. 
The second and third set of constraints encode both the conditions of property \ref{strong_goodness_restate}. By relaxing the 
integrality constraints$(4)$ we obtain a LP relaxation for this problem, LP1 which can be solved in polynomial time to get a 
fractional solution $\bar{x}$.

We round $\bar{x}$ to get the string $s_{r+1}$ by rounding each index, $j$, independently at random treating the 
values of $\bar{x}^j_c$(for a fixed $j$) as the probability of setting the $j^{th}$ character of $s_{r+1}$ to $c$ 
i.e. $s_{r+1}[j]$ is set to $c$ with probability $\bar{x}^j_c$. 

There two possible ways by which this algorithm can fail. The algorithm will fail if IP1 does not have a valid solution, or
if the solution produced after rounding does not satisfy IP1. In lemma \ref{gamma_theorem} and \ref{prob_lemma} we argue that
both these cases happen with very low probability. 

For the rest of the proof we set values for the variables $\alpha = \left\lceil \gamma / |{\cal A}|\right\rceil, \beta = \left\lceil\gamma / {|{\cal A}|}^{2}\right\rceil$ and let $|{\cal A}| = ck$, for some large constant $c$. In the lemma \ref{gamma_theorem} we will derive a bound on the value of $\gamma$ so that IP1 has
a feasible solution for the values of $\alpha,\beta,\delta$ and $|{\cal A}|$ mentioned above. 

\begin{lemma}
\label{gamma_theorem}
IP1 has a solution in the $n^{th}$ (hence all previous iterations), if $\gamma = \Omega(k^2\log n)$.
\end{lemma}

\begin{proof}
We prove this by looking at the agnostic setting by counting the number of labels that have been 
deemed invalid by the labels that have already been chosen in the previous $n-1$ iterations. As long 
as this number is less than the total number possible labels IP1 would have a solution. 

Let us upper bound the number of labels that violate the second set of constraints for some already chosen 
label. i.e. we wish to count the number of labels that intersect with an already chosen label, say $s \in S_{n-1}$, at less 
than $2\alpha$ indices. The number of labels that agree with $s$ at $d$ places is 
given by ${\gamma \choose d} \cdot \left(|{\cal A}| - 1\right)^{\gamma - d}$, where $d$ can take values upto $2\alpha$.
Since there are $n-1$ possible choices for $s$, at most $\Delta_1 = O\left(n \alpha  
{\gamma \choose d} \left(|{\cal A}| - 1\right)^{\gamma - d} \right)$ labels could have been 
deemed infeasible owing to the violation of the second type of constraints.

Now, let us upper bound the number of labels that violate only the third set of constraints. Let $s_1,s_2 \in S_{n-1}$
be a pair of labels that agree on $d_{12}$ indices. Consider
a string $s_3$ which violates the third type of constraint in IP1 for this pair but does not violate the second set 
of constraints for either of the strings. Suppose $s_1, s_2, s_3$ agree at $d_{123}$ indices. Let the number of indices
where $s_3$ agrees with $s_1$ but not $s_2$ be $d_{13}$, similarly define $d_{23}$ to be the number of indices 
where $s_3$ agree with $s_2$ but not $s_1$. By our assumptions, $d_{12}, d_{123} + d_{13}, d_{123} + d_{23} \geq 2\alpha $  and $d_{123} \leq 2\beta/3$. 

One can check that the total number of labels that have been deemed infeasible by the pair $(s_1,s_2)$, is given by
${{\gamma -d_{12}}\choose{d_{13}}} {{d_{12}}\choose{d_{123}}} {{\gamma -d_{12} -d_{13}}\choose{d_{23}}} \left( |{\cal A}| - 1 \right)^{d_{12} - d_{123}} \left( |{\cal A}|-2 \right)^{\gamma - d_{12} - d_{13} - d_{23}}$. 

Since $d_{123},d_{12},d_{23},d_{13}$ can take values upto $2\alpha$ the number of labels that have been deemed infeasible by $s_1$ and $s_2$ is at most 

\[ O\left( \alpha^4 {{\gamma -d_{12}}\choose{d_{13}}}  {{d_{12}}\choose{d_{123}}} {{\gamma -d_{12} -d_{13}}\choose{d_{23}}} \left( |{\cal A}| - 1 \right)^{d_{12} - d_{123}}  \left( |{\cal A}|-2 \right)^{\gamma - d_{12} - d_{13} - d_{23}} \right)\]

Since there are ${n-1 \choose 2}$ ways to select $s_1$ and $s_2$ the total number of labels that have been deemed infeasible, owing
to the violation of only the third type of constraints is \[ \Delta_2 = O\left(n^2 \alpha^4 {{\gamma -d_{12}}\choose{d_{13}}} {{d_{12}}\choose{d_{123}}}  {{\gamma -d_{12} -d_{13}}\choose{d_{23}}}  \left( |{\cal A}| - 1 \right)^{d_{12} - d_{123}}  \left( |{\cal A}|-2 \right)^{\gamma - d_{12} - d_{13} - d_{23}} \right) \]

Now to estimate the maximum value of $\Delta_1 + \Delta_2$ we maximize each of them separately. 
To maximize $\Delta_1$ we calculate the optimal value of $d$ by solving $\frac{\partial \Delta_1}{\partial d} = 0$, i.e. $\frac{\gamma - d}{d} = |{\cal A}| - 1$, which gives $d = \gamma / |{\cal A}|$. Note that this is the expected number of places where $s$ and $s_n$ would have agreed if each character in $s_n$ was chosen uniformly at random from ${\cal A}$.

To maximize $\Delta_2$ we calculate the values of the variables by solving $ \nabla \left[ \log (\Delta_2) \right] =0$ which gives us the following equations.
\begin{eqnarray}
\frac{\gamma - d_{12}-d_{13}}{\gamma - d_{12}} \cdot \frac{d_{12}}{d_{12} - d_{123}} \cdot \frac{\gamma - d_{12}-d_{13} - d_{23}}{\gamma - d_{12} - d_{13}}\cdot\frac{|{\cal A}| -1}{|{\cal A}| -2} &=& 1  \nonumber \\
\frac{\gamma - d_{12} - d_{13} }{d_{13}} \cdot \frac{\gamma - d_{12} - d_{13} - d_{23} }{\gamma - d_{12} - d_{13}} \cdot \frac{1}{|{\cal A}| -2}&=& 1 \nonumber \\
\frac{\gamma - d_{12}-d_{13} - d_{23}}{d_{23}} \cdot \frac{1}{|{\cal A}| -2} &=& 1  \nonumber \\
\frac{d_{12}-d_{123}}{d_{123}} \cdot \frac{1}{|{\cal A}| -1} &=& 1  \nonumber
\end{eqnarray}
Solving for $d_{12}, d_{13}, d_{23}$ and $d_{123}$ we find that $\Delta_2$ is maximized when each of the variables is equal to 
their expected values if $s_n$ was chosen randomly from ${\cal A}^{\gamma}$, i.e. $d_{12} \approx d_{13} \approx d_{23} \approx \gamma/|{\cal A}|$ and $d_{123} \approx \gamma/|{\cal A}|^2$. Substituting the values of the variables and simplifying by sterling's approximation, we get that $\Delta_1 + \Delta_2$ is upper bounded by $O(n^2  |{\cal A}|^{\gamma} {\zeta}^{-\frac{\gamma}{|{\cal A}|^2}})$ for a fixed constant $\zeta > 1$. 

IP1 will have a solution as long as the number of infeasible strings is less than the total number of strings, i.e. 
\begin{eqnarray}
\Delta_1 + \Delta_2 &<& |{\cal A}|^{\gamma} \nonumber \\
O(n^2  |{\cal A}|^{\gamma} {\zeta}^{-\frac{\gamma}{|\cal A|^2}}) &<& |{\cal A}|^{\gamma} \nonumber \\
O\left(|{\cal A}|^2\log n\right) &<& \gamma \nonumber
\end{eqnarray}

Thus by setting $\gamma$ to be $\zeta' |{\cal A}|^2(\log n)$, for some large enough constant $\zeta' > 1$, we can be sure
that IP1 has a solution in each of the first $n$ iterations. 
\end{proof}

The following lemma bounds the probability of the event that the label produced by randomized rounding does not satisfy IP1.
\begin{lemma}
\label{prob_lemma}
The label obtained after rounding satisfies property \ref{strong_goodness_restate} with probability $1 - o(1)$.
\end{lemma}
\begin{proof}
In order to check if the rounded solution is feasible we have to verify if $s_{p+1}$ satisfies IP1. The 
first set of constraints are trivially satisfied. For a fixed $s_p \in S_r$ the expected number of 
places where $s_p$ agrees with $s_{r+1}$ is lower bounded by $2\alpha$, i.e. $E\left[ |s_p \odot s_{r+1}|\right] \geq 2\alpha$. 
By chernoff bounds\cite{chernoff}, $|s_p \odot s_{r+1}|$ will be less than $\alpha $ with probability at most 
$e^{-\alpha/2}$. By union bound the probability that for at least one $s_p \in S_{n-1}$, $|s_p \odot s_{n}|$ is 
less than $\alpha$ is at most $n e^{-\alpha/2}$.
Similarly at least one of the third set of constraints is violated with probability at most 
$\frac{n^2}{2} e^{-\beta/6}$.  We conclude that the probability of failure
is at most $n e^{-\alpha/2} + \frac{n^2}{2} e^{-\beta/6}$ which is less than $n^2e^{-\beta/6} = n^2e^{-\gamma/{6|{\cal A}|}^2} $. 
Substituting values for the variables we get the failure probability is $\frac{1}{n^{(\zeta'-12)/6 }}$ which is 
$o(1)$ for large enough constant $\zeta'$. 
\end{proof}

Using the algorithm shown above we can construct, in polynomial time, a bipartite graph $G = (L \cup R, E)$, which satisfies $(\alpha,\beta )$-strong goodness property with high probability, such that $|L| = n$, $|R| = \gamma|{\cal A}| = O(k^3 \log n)$. Since $ \alpha / \beta = \frac{\gamma/|{\cal A}|}{\gamma/|{\cal A}|^2} = |{\cal A}| > k$, observation \ref{link_properties} along with theorem \ref{main_thm} yields an $O(k^3 \log n)$ approximate algorithm for VC-SNDP. Next, we present a deterministic version of the above randomized algorithm.

\subsection{Derandomizing the Algorithm}
\label{derand_algo}
Our derandomized algorithm runs in $n$ iterations each consisting of a number of steps. 
As with the randomized algorithm presented earlier, our derandomized algorithm iteratively builds a set of feasible labels that satisfy 
property \ref{strong_goodness_restate}. 
In each iteration our algorithm uses local search to find a 
label to augument the set of feasible labels that were chosen earlier. Starting with an arbitrary(possibly infeasible) label, in
each step of the local search, our algorithm changes one character of the current label based on a potential function, to
move \textit{closer} to a feasible solution. We show that for an appropriately chosen potential function we can reach a feasible
label in polynomially bounded number of steps. 

We use the following potential function to guide our local search. For any label $s \in {\cal A^{\gamma}}$ define 
\[ \phi(s) = \sum_{s_i \in S_r} max\left\{0, \alpha - |s_i \odot s| \right\}  \ + \
\sum_{s_j \in S_r} \sum_{s_i \in S_r, s_i \neq s_j} max\left\{0, |s_i \odot s_j \odot s| - \beta \right\} \]

In every step of the local search we use local operations(changing one character) to move to a label that strictly 
improves that value of the above potential function. In lemma \ref{local_lemma} we show that this is always possible, i.e. one can 
always improve the value of the potential function by just changing a single character. Since $\phi$ is bounded by $n\alpha + {n \choose 2} \beta$
and we achieve a feasible solution once the potential function drops to $0$, this is suffices to 
show that each iteration terminates with a feasible label in polynomial number of steps.

\begin{lemma}
\label{local_lemma}
For every label $\omega \in {\cal A}^{\gamma}$ either $\phi(\omega) = 0$, or there exists another label 
$\omega'$ such that $\phi(\omega')$ is less than $\phi(\omega)$ and $\omega$ and $\omega'$ differ at exactly one index.
\end{lemma}

We will now present some definitions that will be useful in proving lemma \ref{local_lemma}. Let $S_r$ be the set of labels 
chosen after the first $r^{th}$ iterations of the algorithm. Consider the following integer program that represents the
set of feasible labels that can be used to augment $S_r$.
\begin{eqnarray} 
    (IP2) \hspace{3 cm} \label{alloc_eqn} \sum_{c \in \cal A} x^j_c &=& 1 \hspace{4.3 cm} \forall j \in [\gamma] \\ 
    \label{alpha_eqn} \sum_{j \in [\gamma]} \sum_{c \in \cal A} x^j_c \cdot f(s_p,j,c) &\geq& \alpha  \hspace{4.2 cm} \forall p \in S_r \\
    \label{beta_eqn} \sum_{j \in [\gamma]} \sum_{c \in \cal A} x^j_c \cdot g(s_p, s_q,j, c) &\leq& \beta  \hspace{4.3 cm} \forall p,q\in S_r\\
    \label{x_eqn} x^j_c &\in& \left\{ 0,1\right\} \hspace{3.5 cm} \forall j \in [\gamma], c \in \cal A
\end{eqnarray}

Let LP2 be the standard linear programming relaxation of IP2 and let ${\cal P}$ represent the polytope of the feasible solutions to LP2.
Let ${\cal P'}$ be the polytope corresponding to the sets of equations (\ref{alloc_eqn}), (\ref{x_eqn}) for LP2. 
Clearly ${\cal P}$ is contained in ${\cal P'}$. By lemma \ref{prob_lemma} proved earlier, ${\cal P}$ is feasible, and infact shares 
a large number of its vertices with ${\cal P'}$.

We will require the following claim to complete the proof of lemma \ref{local_lemma}. 
\begin{claim}
\label{plane_claim}
No plane corresponding to an equation chosen from the sets (\ref{alpha_eqn}), (\ref{beta_eqn})
intersects an edge in ${\cal P'}$.
\end{claim}

\begin{proof}
We begin by noting that ${\cal P'}$ is the convex hull of characteristic vectors of labels in ${\cal A}^{\gamma}$ 
and that two vertices in ${\cal P'}$ are connected by an edge iff their corresponding labels differ at exactly one position. 
If possible let there be an inequality, $\hat{e}$, in the 
sets (\ref{alpha_eqn}), (\ref{beta_eqn}) such that the plane corresponding to it intersects an edge of ${\cal P'}$. 
Thus we have two adjacent corner points $v$ and $v'$, such that $\lambda v + (1-\lambda)v'$ causes 
$\hat{e}$ to be satisfied with equality for some positive value of $\lambda \in (0,1)$. Moreover, $\lambda v + (1-\lambda)v'$
is integral except for two variables say $x^j_c$ and $x^j_{c'}$. Notice that both the variables cannot appear in $\hat{e}$. 
Exactly one of these two variables must a have non-zero coefficient(equal to $1$)
in $\hat{e}$. This is because if neither of the them appears in $\hat{e}$ then 
moving from $v$ to $v'$ along the edge connecting them should not change the feasibility of the 
solution. Since exactly one of the fractional variables in $\lambda v + (1-\lambda)v'$ appears in $\hat{e}$, 
and the right hand side of the equation is integral, $\lambda v + (1-\lambda)v'$ cannot 
satisfy $\hat{e}$ with equality. 
\end{proof}

Armed with these definitions we now give the proof of lemma \ref{local_lemma}. 

\begin{proof}[\textbf{Proof of lemma \ref{local_lemma}}]
Throughout this proof, by a slight abuse of notation, we will use $\phi(v_s)$ to denote $\phi(s)$ 
for any label $s \in {\cal A}^{\gamma}$ and its corresponding $v_{s}$.
We extend the domain of $\phi$ to fractional points in the interior of ${\cal P'}$ by defining 
\[ \phi(x) = \sum_{s_i \in S_r} max\left\{0, \alpha - x \cdot x_{s_i} \right\}  + 
\sum_{s_j \in S_r} \sum_{\ s_i \in S_r, s_i \neq s_j} max\left\{0, x_{s_i} \cdot x_{s_j} \cdot x - \beta \right\} \]

Here $x_s$ represents the characteristic vector for label $s$ and for
any tripple of vectors $(x,y,z)$, $x\cdot y \cdot z = \sum_{j \in [\gamma]} x_jy_jz_j$. Let $v_{\omega} \in {\cal P'}$ be the 
vertex corresponding to the label $\omega$ and let 
$N_{\omega}$ be the set of labels that differ with ${\omega}$ at exactly one position. If $v_{\omega} \in {\cal P}$,
then $\phi({\omega}) =0$ and we are done. Let us assume $v_{\omega} \notin {\cal P}$ i.e. $\phi({\omega}) > 0$. 
By lemma \ref{prob_lemma}, ${\cal P}$ must have a feasible integer point which should correspond to a vertex 
of ${\cal P'}$, say $v^*$, so $\phi(v^*) = 0$. 
For $\lambda >0$ define $v_{\lambda} = (1-\lambda) v_{\omega} + \lambda v^*$. As we increase $\lambda$
from $0$ to $1$ we travel along the vector $\overline{v_\omega v^*}$. 
Thus, $\left[\frac{\partial \phi(v_{\lambda})}{ \partial \lambda} \right]_{\lambda = 0}< 0$. 
By convexity of ${\cal P'}$, the 
vector $\overline{v_{\omega}v^*}$ can be written as a convex combination of vectors $\overline{v_{\omega}u}$ 
for $u \in N_{\omega}$. Thus by moving in at least one of the directions $\overline{v_{\omega}u}$ we reduce the 
potential value. By claim \ref{plane_claim}, the potential $\phi(v_\lambda)$ will continue to decrease as 
we move along this direction i.e. it will not happen that potential falls and then rises 
as we move along this edge. Thus at least one vertex in $N_{\omega}$ must have lower potential than $v_\omega$.
\end{proof}

%
%

\section{Deterministic Algorithm for Single-Source VC-SNDP}
In the single-source VC-SNDP we are given a $G = (V,E)$ with a special vertex $s$
called the source, and a subset $T$ of vertices called terminals. Additionally, for each terminal $t \in T$
we are given a connectivity requirement $r(s, t) \leq k$. We are required to select a minimum cost
subset of edges $E'$ such that inthe graph induced by $E'$ every terminal t has $r(s, t)$ vertex-disjoint 
paths to $s$. As before, we create a family $\left\{T_1, \cdots , T_m\right\}$ of subsets of
terminals and also create $m$ copies of the graph $G_1, \cdots, G_m$, and for each copy we
solve the element-connectivity SNDP instance with connectivity requirements induced by
terminals in $T_i$. Let $E_i$ be the 2-approximate solution to instance $G_i$. The final solution is the union of the solutions for all the instances, i.e. $\bigcup_{i \in [m]} E_i$. Clearly, the cost of this solution is at most $2m$ times $OPT$.

\begin{definition}[Good Family of Subsets]
Let \textsl{M} be the input collection of source-sink pairs and $T$ is the corresponding
collection of terminals. We say that a family $\left\{T_1, \cdots , T_m \right\}$ of 
subsets of T is \textit{good} iff for each terminal $t \in T$, 
for each subset of terminals, $X$, not containing $t$, and of size at most $(k - 1)$, there
is a subset $T_i$, such that $t \in T_i$ and $X\cap T_i  = \phi$.
\end{definition}

In the same vein as theorem \ref{main_thm} we now show that the existance of such a family 
would imply that the output of the above algorithm is a feasible solution to VC-SNDP.

\begin{theorem}[\cite{khanna}] 
\label{main_thm_2}
Let $\left\{T_1, \cdots, T_m\right\}$ be a good family of subsets. 
Then the output of the above algorithm is a feasible solution to the single-source VC-SNDP instance.
\end{theorem}
\begin{proof}
Let $t \in T$ be a terinal and let $X$ be a subset of $V \ {s, t}$ of size at most $k-1$. 
It is suffices to prove that the removal of $X$ from the graph
induced by $E'$ does not disconnect $s$ and $t$. Let $X' = X \cap T$. Since we start with 
a good family, there is some $T_i$ which contains $t$ and does not intersect with $X'$. 
Let $E_i$ be the solution to the corresponding $k$-element connectivity instance. 
Since vertices of $X$ are non-terminal vertices for the instance $G_i$, their removal from the graph 
induced by $E_i$ does not disconnect $s$ from $t$.
\end{proof}

Now we present a deterministic algorithm
to construct such a family of size $O(k^2 \log n)$ which gives us an $O(k^2 \log n)$
approximate algorithm. 

\section{Algorithm to Construct Good Family of Subsets}
As before we can view a good family of subsets as a bipartite graph, $G = (L \cup R, E)$ and restate the 
definition of a good family of subsets may be restated as follows. 

\begin{property}[Weak Goodness]
\label{goodness_property_2}
A graph $G$ satisfies the \textit{weak goodness property} if for every $\ell_i \in L$ and 
$X \subseteq L$ such that $|X| < k$ and $\ell_i \notin X$, $N(\left\{\ell_i\right\})$ is not contained in  $N(X)$.
\end{property}

We use the following alternate definition of good family of subsets in our algorithm.

\begin{property}[Strong Goodness]
\label{strong_goodness_2}
A graph $G = (L \cup R, E)$ satisfies the $(\alpha, \beta)-strong \ goodness$ property if 
\begin{enumerate}
\item for all $\ell_i \in L$,  $|N(\left\{\ell_i\right\})| = \alpha$
\item for all distinct $\ell_i, \ell_j \in L$,  $|N(\left\{\ell_i\right\})\cap N(\left\{\ell_j\right\})| \leq  \beta$
\end{enumerate}
\end{property}

The following observation follows easily from the definitions above.

\begin{observation}
\label{link_properties_2}
If a graph $G$ satisfies the \textit{$(\alpha,\beta)$-strong goodness} property for $\alpha / \beta  > k$,
 then it satisfies the weak goodness property.
\end{observation}

Now we describe a deterministic algorithm to construct a graph 
$G$ such that it satisfies the \textit{$(\alpha,\beta)$-strong goodness} property
for $\alpha / \beta > k$. The variables $\alpha,\beta$ will be determined later.

\subsection{Algorithm for Constructing Graphs with Strong Goodness Property}
\label{algo_section_2}
Our algorithm constructs a graph $G=(L \cup R, E)$, for $|L| = n$ by assigning a string $w_i \in \cal A^{\gamma}$
to every $\ell_i \in L$. There are $|\cal A|\gamma$ vertices in $R$
which are indexed by $[\gamma] \times \cal A$. There is an edge connecting $\ell_i$ to $(j,c)$ iff 
$w_i[j] = c$, i.e. the $j^{th}$ character of $w_i$ is $c$. Property \ref{strong_goodness_2} can be restated as follows,

\begin{property}
\label{strong_goodness_restate_2}
A graph $G$ satisfies the \textit{$(\alpha, \beta)$-strong goodness} property if 
\begin{enumerate}
\item for all $\ell_i \in L$,  $|w_i| = \alpha$
\item for all distinct $\ell_i, \ell_j \in L$,  $|w_i \odot w_j | \leq \beta$
\end{enumerate}
\end{property}

In describing our algorithm we will set the variables $\alpha$ and $\beta$ as follows, $\alpha = \gamma$
and $\beta = \gamma/|{\cal A}|$. Our algorithm builds the set of labels by solving the following integer program in every iteration. 
\begin{eqnarray} 
    (IP3) \hspace{3 cm} \sum_{c \in \cal A} x^j_c &=& 1 \hspace{4.8 cm} \forall j \in [\gamma] \\ 
    \sum_{j \in [\gamma]} \sum_{c \in \cal A} x^j_c \cdot f(s_p,j,c) &\leq& 2\beta/3  \hspace{4.2 cm} \forall p \in S_r \\
    x^j_c &\in& \left\{ 0,1\right\} \hspace{4.0 cm} \forall j \in [\gamma], c \in \cal A
\end{eqnarray}

Lemma \ref{gamma_theorem_2} given below establishes the minimum value of $\gamma$ for which the above integer program has a solution in
each of the $n$ iterations. 

\begin{lemma}
\label{gamma_theorem_2}
IP3 has a solution in the $n^{th}$ (hence all previous iterations), if $\gamma = \Omega(k\log n)$.
\end{lemma}
\begin{proof}
Omitted.
\end{proof}

Let LP3 be the linear programming relaxation for IP3 and let $\bar{x}$ be a (fractional)feasible solution to LP3.
We can round $\bar{x}$ to get an integer solution by 
rounding each index, $j$, independently at random treating the 
values of $\bar{x}^j_c$(for a fixed $j$) as the probability of setting the $j^{th}$ character of $s_{r+1}$ to $c$ 
i.e. $s_{r+1}[j]$ is set to $c$ with probability $\bar{x}^j_c$. 


The following lemma bounds the probability of the event that the label produced by randomized rounding does not satisfy property \ref{strong_goodness_restate_2}.
\begin{lemma}
\label{prob_lemma_2}
The label obtained after rounding satisfies property \ref{strong_goodness_restate_2} with probability $1 - o(1)$.
\end{lemma}
\begin{proof}
Omitted.
\end{proof}

Using the algorithm shown above we can construct, in polynomial time, a bipartite graph $G = (L \cup R, E)$, which satisfies $(\alpha,\beta )$-strong goodness property with high probability, such that $|L| = n$, $|R| = \gamma|{\cal A}| = O(k^2 \log n)$. Since $ \alpha / \beta = |{\cal A}| > k$, observation \ref{link_properties_2} along with theorem \ref{main_thm_2} yields an $O(k^2 \log n)$ approximate algorithm for VC-SNDP. 

\subsection{Derandomizing the Algorithm}
The derandomized algorithm is similar to the one presented earlier in section \ref{derand_algo}. The algorithm 
proceeds in several iterations each consisting of multiple steps. In each iteration it 
augments the set of feasible labels. It finds such a feasible solution by using local seach starting 
from an arbitrary label. In every step of the local search it reduces a 
polynomially bounded potential function until it falls to
zero, in which case we obtain a label satisfying property \ref{strong_goodness_restate_2}.

We use the following potential function to guide our local search. For any label $s \in {\cal A^{\gamma}}$ define,
\[ \phi(s) = \sum_{s_i \in S_r} max\left\{0, |s_i \odot s| - \beta \right\} \]

In lemma \ref{local_lemma_2} we show that one can 
always improve the value of the potential function by just changing a single character of the current label. 
Since $\phi$ is bounded by $n\beta$
and we achieve a feasible solution once the potential function drops to $0$, this is suffices to 
show that each iteration terminates with a feasible label in polynomial number of steps.

\begin{lemma}
\label{local_lemma_2}
For every label $\omega \in {\cal A}^{\gamma}$ either $\phi(\omega) = 0$, or there exists another label 
$\omega'$ such that $\phi(\omega')$ is less than $\phi(\omega)$ and $\omega$ and $\omega'$ differ at exactly one index.
\end{lemma}
\begin{proof}
Similar to the proof for lemma \ref{local_lemma}.
\end{proof}


%
%
%

\section*{Acknowlegement}
I would like to thank Vijay Vazirani for his valuable guidance and advice, and for being a sounding board for
my ideas. I would also like to thank Lei Wang for his valuable inputs and suggestions at various stages of this 
paper.

\bibliography{VCSNDP}
\bibliographystyle{plain}

\end{document}